\documentclass[conference]{IEEEtran}
\ifCLASSINFOpdf
\else
\fi
\usepackage{amssymb}
\usepackage{amsmath,amsthm}
\usepackage{setspace}
\usepackage{tabularx}
\usepackage{booktabs}
\usepackage{longtable}
\usepackage{graphicx}
\usepackage{adjustbox}
\usepackage{soul,color}
\usepackage{lipsum}
\usepackage [english]{babel}
\usepackage [autostyle, english = american]{csquotes}
 \usepackage{verbatim}
 \usepackage{comment}
 \usepackage{longtable}
\usepackage{array,arydshln,tabu}
\usepackage{cite}
\MakeOuterQuote{"}
\usepackage{xcolor,soul,verbatim}
\usepackage[colorlinks=true]{hyperref}
\newtheorem{theorem}{Theorem}[section]

\newtheorem{definition}[theorem]{Definition}

\newcommand{\hlc}[2][yellow]{{%
    \colorlet{foo}{#1}%
    \sethlcolor{foo}\hl{#2}}%
}

\newcommand{\Fq}{\mathbb{F}_q}
\newcommand{\Z}{\mathbb{Z}}
\newcommand{\qrn}{{\mathbb{F}_q[x] \over \langle x^n-1\rangle}}
\newcommand{\qrm}{{\mathbb{F}_q[x] \over \langle x^m-a\rangle}}
\addto\captionsenglish{}

\addto\captionsenglish{}

\begin{document}
%
\title{New Linear Codes as Quasi-Twisted Codes from Long Constacyclic Codes}

\author{\IEEEauthorblockN{Nuh Aydin}
\IEEEauthorblockA{Department of Mathematics and Statistics\\
Kenyon College\\
Gambier, OH  43022\\
Email: aydinn@kenyon.edu}
\and
\IEEEauthorblockN{Thomas Guidotti}
\IEEEauthorblockA{Kenyon College\\
Gambier, OH  43022\\
Email: guidotti1@kenyon.edu}
\and
\IEEEauthorblockN{Peihan Liu}
\IEEEauthorblockA{Kenyon College\\
Gambier, OH  43022\\
Email: liu4@kenyon.edu}}


%


\maketitle

\begin{abstract}
One of the most important and challenging problems in coding theory is to determine the optimal values of the parameters of a linear code and to explicitly construct codes with optimal parameters, or as close to the optimal values as possible. The class of quasi-twisted (QT) codes has been very promising in this regard. Over the past few decades various search algorithms to construct QT codes with better parameters have been employed. Most of these algorithms (such as ASR \cite{qtmain}) start by joining constacyclic codes of smaller lengths to obtain QT codes of longer lengths. There has been an algorithm that works in the opposite way that constructs shorter QT codes from long constacyclic codes. We modified and generalized this algorithm and obtained new linear codes via its implementation. We also observe that the new algorithm is related to the ASR algorithm.       
\end{abstract}


\ifCLASSOPTIONpeerreview
 \begin{center} \bfseries EDICS Category: 3-BBND \end{center}
\fi
%
\IEEEpeerreviewmaketitle

\section{Introduction}

A linear block code $C$ of length $n$ over the finite field $\Fq$ is a vector subspace of $\Fq^n$. The elements of $C$ are called codewords.  If the dimension of $C$ is $k$, then it is referred to as an $[n,k]_q$-code. If the minimum (Hamming) distance (weight) is $d$, then and it is an $[n,k,d]$-code. One of the most important and challenging problems in coding theory is to determine the optimal values of the parameters $n,k,d$, given the alphabet size $q$ and explicitly construct codes whose parameters attain the optimal values. The problem can be formulated in a few different ways. One common version is to fix $q,n$ and $k$ and look for the maximum possible value of $d$. In general, this optimization problem is very hard and in most cases codes with optimal parameters  are not known.
The online database \cite{database} contains data about what is known about this problem for codes over the alphabets $\Fq$ for $q\leq 9$ up to certain length for each alphabet. In most cases, there are gaps between the minimum distances of best know linear codes (BKLC) and the best theoretical upper bound on $d$.     

This optimization problem is hard for two main reasons. First, the number $\displaystyle{\frac{(q^n-1)(q^n-q)\cdots (q^n-q^{k-1})}{(q^k-1)(q^k-q)\cdots (q^k-q^{k-1})}}$  of linear codes  over $\Fq$ of length $n$ and dimension $k$  is large and grows fast. Hence exhaustive computer searches are not feasible for all but small values of $n$ or $k$. Secondly, computing the minimum distance (weight) of a linear code is computationally intractable \cite{NPHard}. For most entries in the database \cite{database}, optimal values of the parameters are not attained. Optimal codes are generally known when either $k$ or $n-k$ is small. 

Numerous approaches, techniques, and search methods have been employed to improve the parameters of BKLCs to get closer to the the optimal values. It is unlikely that a single method will work to solve most instances of this challenging problem. One method that has been quite effective to obtain new codes has been computer searches in the class of quasi-twisted (QT) codes. The algorithm ASR introduced in \cite{qtmain} is one such algorithm. It has been improved in recent years and  produced dozens of record breaking codes. The ASR algorithm searches for new linear codes using a special type of 1-generator QT codes that have generators in a particular form.  It starts with a short constacyclic code and uses it as a building block  to construct longer QT codes.  Another method was introduced in \cite{tdoriginal} that works in the opposite direction, that is by starting with a very long constacyclic code and obtains shorter QT codes from it. In a recent paper \cite{Chen2015}, new linear codes have been obtained via a modification of the original method. In this work, we modified and generalized this method and obtained new linear codes from its implementation. Finally, we notice a connection between the new method and the ASR algorithm.

\section{Preliminaries} 
Cyclic codes are one of the most important classes of codes in algebraic coding theory for both theoretical and practical purposes. They are extensively studied and generalized in many directions. Some of the well known generalizations of cyclic codes are constacyclic codes, quasi-cyclic (QC) codes, and quasi-twisted (QT) codes. 

\begin{definition}
Let $\Fq$ be the finite field with $q$ elements and let $a\in \Fq^{*}=\Fq\setminus \{0\}$. A linear code $C$ of length $n$ over $\Fq$ is called a quasi-twisted (QT) code of index $\ell$ (or an $\ell$-QT code) if it is closed under the constacyclic shift by $\ell$ positions, i.e, for any codeword  $\bold{c}=(c_0,c_1,\cdots,c_{n-1}) \in C$, we also have $\pi_{\ell,a}(\bold{c})=(a c_{n-\ell}, \cdots,$ $ac_{n-1}, c_0, c_1,\cdots,$  $c_{n-\ell-1}) \in C$. 
\end{definition}

The smallest such positive integer $\ell$ is called the index of $C$ and it must divide the code length $n$. Hence, $n=m\cdot \ell$ for some $m\in \Z^{+}$. The scalar $a\in \Fq^{*}$ is called the shift constant. The following are some of the most important special cases of QT codes:
\begin{itemize}
\item $a=1, \ell=1$  gives cyclic codes
\item $a=-1, \ell=1$  gives negacyclic codes
\item $a=1$  gives quasi-cyclic (QC) codes
\item $\ell=1$  gives constacyclic codes
\end{itemize}

One of the reasons why cyclic codes are so prominent in coding theory is they establish a key link between algebra and coding theory through the correspondence between vectors $\bold{v}=(v_0,v_1,\cdots,v_{n-1})$ and polynomials\\ $v(x)=v_0+v_1x+\cdots+v_{n-1}x^{n-1}.$ This map establishes a vector space isomorphism between $\Fq^n$ and $\Fq[x]_{<n}=\{p(x) \in \Fq[x]: \deg(p(x)) <n \}$, the set of all polynomials of degree $<n$ over $\Fq$.  It is well known that under this identification cyclic codes of length $n$ over $\Fq$ correspond to the ideals of the quotient ring $\qrn$. Under the same identification, the algebraic structure of a QT code of length $n=m\cdot \ell$ is an $R$-module of $R^{\ell}$, where $R=\qrm$. If $C$ is generated by $r$ elements of $R^{\ell}$ then it is called an $r$-generator QT code. A generator matrix of an $r$-generator QT code can be put, by applying a suitable permutation of the columns if necessary, into the form: 
$$
\begin{bmatrix} 
G_{11} & G_{12} & \cdots & G_{1l}  \\
G_{21} & G_{22} & \cdots & G_{2l} \\
\cdots & \cdots & \cdots & \cdots & \\
G_{r1} & G_{r2} & \cdots & G_{rl} \\
\end{bmatrix}
$$
where each $G_{ij}$ is an $a$-circulant (also called a twistulant) matrix of the form

$$
\begin{bmatrix} 
c_0&c_1& c_2 & \cdots    &c_{m-1}\\
ac_{m-1}&c_0&c_1&\cdots &c_{m-2}\\
ac_{m-2}&ac_{m-1}&c_0&\cdots &c_{m-3}\\
\hdots & \hdots  & \hdots & \hdots & \hdots  \\
\end{bmatrix}
$$

\noindent where each row is a constacyclic shift of the previous row.



The ASR search algorithm that is introduced in \cite{qtmain} is based on the following theorem.
\begin{theorem}\label{theorem:ASR}
\cite{qtmain} Let C be a 1-generator QT code of length $n = m\ell$ over $\Fq$
with a generator of the form
$$(g(x)f_1(x), g(x)f_2(x), . . . , g(x)f_{\ell}(x))$$ where $x^m - a = g(x)h(x)$ and $\gcd(h(x),f_i(x)) = 1$ for all $i = 1,...,\ell$. Then $dim(C) = m -deg(g(x))$, and $d(C) \geq \ell\cdot d$ where $d$ is the minimum distance of the constacyclic code $C_g$ generated by $g(x)$.
\end{theorem}
This algorithm has been refined and automatized  in more recent works such as (\cite{gf2}, \cite{gf7}, \cite{gf3}, \cite{nonprime}, \cite{involve}) and dozens of record breaking codes have been obtained over every finite field $\Fq$, for $q=2,3,4,5,7,8,9$ through its implementation. Moreover, it has been further generalized in \cite{genASR} and more new codes were discovered  that would have been missed by its earlier versions.

In an implementation of the ASR algorithm, we begin by choosing the alphabet  $\Fq$ and the shift constant $a$. We then pick a desired block length $m$ and number of blocks $\ell$ such that the resulting length of the code is $n = m\cdot \ell$.  Next, we find all divisors of the polynomial $x^m-a$ and end up with a list of possible generator polynomials $g(x)$ (with corresponding check polynomial $h(x)$ such that $g(x)\cdot h(x) = x^m - a$) of varying degrees, each generator giving a constacyclic code $C_g$ of length $m$ with dimension $k = m -\deg(g(x))$. Then the search program generates polynomials $f_{i}(x)$ where $\gcd(f_{i}(x), h(x)) = 1$ for all $i$ $1 \leq i \leq \ell$, and for each set $\{f_1(x),\dots,f_{\ell}(x) \}$ of polynomials constructs the QT code $C$ with a generator of the form given in Theorem~\ref{theorem:ASR}. Note that, each block is a constacyclic code generated by $g(x)f_i(x)$. In fact, $\langle g(x)f_i(x)\rangle = \langle g(x)f_j\rangle$ for all $i,j$ under the condition   $\gcd(h(x),f_i(x)) = 1$.
We simply join all of these $\ell$ blocks together to obtain a QT code whose minimum distance is guaranteed to be at least $\ell\cdot d$. In many cases, its actual minimum distance is much bigger.


There is another method of constructing QC and QT codes that works in the opposite way (which we  may informally refer to as "top-down method"). The basic idea of this method goes back to \cite{td91} and \cite{tdoriginal}. It starts with a long cyclic code of composite length  $N$ and permutes the coordinates of the code so that each block is a cyclic code of length dividing $N$. This method was later refined to search algorithms that produced QT 2-weight codes from constacyclic codes in \cite{Chen2009}

One improvement of the algorithm is given in \cite{Chen2015} which applied the method to the specific case of simplex codes, the duals of Hamming codes. Simplex codes are constacyclic themselves, so the idea of applying a permutation to a long constacyclic code to form a matrix that is a generator for a QT code can be applied here. This version of the algorithm uses the idea of a weight matrix to construct better QT codes. After applying the permutation to the generator of the constacyclic code, a $p \times p$ weight matrix is created to store the weights of the defining polynomials $g_{1}(x), \dots, g_{p}(x)$. The first row of this matrix has the weights $wt(g_{1}(x)), \dots,wt(g_{p}(x))$ and all subsequent rows are its cyclic shifts, so it suffices to store the first row of the matrix only. Because of this construction, it follows that the minimum distance of a QT $[n=t\cdot m, k]_{q}$ code is determined by the minimum row sums of the chosen columns from the weight matrix. So to form a QT code with these parameters having the highest possible minimum distance it suffices to maximize the minimum row sums. This gives rise to the iterative algorithm found in \cite{Chen2015}. The idea is that one will build a QT code from maximizing the minimum row sums one column at a time. So given a QT $[i\cdot m, k]_{q}$ code the algorithm tries to construct a QT $[(i+1)\cdot m \cdot, k]_{q}$ code by using the weight matrix to add another block to the resulting QT code. If the maximum number of blocks desired is $t$,  then the algorithm constructs codes with parameters $[m, k]_{q}, [2\cdot m, k]_{q} \cdots [t\cdot m, k]_{q}$.


\vspace{2mm}

\section{Our Contribution}

In this work, we generalize and modify the search method described in the previous section. As a result of our implementation of the generalized algorithm, we have found 5 new linear codes that improve the bounds in the database \cite{database}. 

The previous version of the  algorithm is restricted to the class of simplex codes. Our first generalization is to  apply it to the broader class of all constacyclic codes. We first choose the alphabet $\Fq$, a shift constant $a\in \Fq^{*}$, and a length $N$ with many factors. Next, we  find a number of generator polynomials $G(x)$ such that $G(x)|(x^N - a)$.  We want the degree of this generator to be relatively high, so our initial  program (written in Magma) that outputs the generator polynomials has constraints on their degrees. After reading in all of these long generator polynomials to a C++ program, we want to find all combinations of $m \cdot p = N$. This is obviously determined by the prime factorization of $N$. We  refer to $m$ as the block length and then $p$ is  the total number of blocks. For a given combination of $(m, p)$ we  first represent our long polynomial $G(x)$ as a vector in the usual way. Then we  split up this vector into $p$ vectors of length $m$ each, such that each vector combines columns $i,i + p, i + 2p, \dots, i + (m-1)p$ for $0 \leq i \leq p-1$. Next, we  convert all of these vectors back to polynomials $g_{i}(x)$ in the usual way and then compute $\gcd(x^m-a, g_{1}(x), g_{2}(x), \dots g_{p}(x))$. Computing this gcd is very similar to the ASR search in which we begin with a standard generator polynomial $g(x)$ (a divisor of $x^m - a$) and then multiplying it by polynomials $f_{i}(x)$ in each block such that $\gcd(h(x), f_{i}(x)) = 1$. The process here is simply the reverse: we are starting with $g(x) \cdot f_{i}(x)$ in each block and we want to determine the standard generator $g(x)$ based on this. Obviously then the dimension of the constacyclic code with standard generator $g(x) =\gcd(x^m-a, g_{1}(x), g_{2}(x), \dots g_{p}(x))$ is  $k = m -\deg(g(x))$. If $k$ is equal to the dimension of the original constacyclic code, then we are dealing with the $1$-generator QT case, otherwise we are dealing with a multi-generator case. Our last step here is to construct a $a$-circulant (twistulant) matrix corresponding to each of these defining polynomials $g_{i}(x)$ for $1 \leq i \leq p$, keeping only those matrices whose rank is equal to $k$ or $k-1$. We observed that using blocks that have the same rank usually gives the best results in terms of obtaining codes with high minimum distances.

Obviously the dimension $k$ has an upper bound of $m$, so it only makes sense to consider combinations of $(m, p)$ that will yield reasonable dimensions. In many cases where there are a high number of defining polynomials, their gcds with $x^m-a$ have a very low degree, resulting in a QT code whose dimension is very close to block length $m$.

For each combination of $(m, p)$ we apply the column permutations as described to put the large circulant matrix into $p$ blocks of length $m$ circulant matrices. We then horizontally join $t$ of these matrices to form the generator matrix for a QT code. An improvement we made to this algorithm deals with the rank of these circulant matrices. In \cite{tdoriginal}, the author notes that in cases where $N - \deg(G(x) \neq m - \deg(g(x))$ ($ G(x)$ is the generator polynomial of the long constacyclic code and $g(x)$ is the gcd of defining polynomials) we are dealing with a multi-generator QT code. Chen suggests that in such a case the smaller dimension generator polynomial can be augmented by adding rows one at a time such that the rows are linearly independent, in this way a code with dimension $ N-\deg(G(x))$ may be constructed. After some testing using this idea we had no success in  finding codes with parameters close to those of  BKLCs, so instead we decided to construct codes with dimension $k=m-\deg(g(x))$ simply taking the blocks of circulant matrices of rank $k$ rather than by  adding additional linearly independent rows.  Additionally we noticed that  amongst all of the  $p$ circulant matrices, there are many matrices of rank $k-1$ so we keep those as well. So once we choose $k$ after finding the defining polynomials, we  first determine if $k$ is  large enough to find reasonable results. After that we go through all circulant matrices and keep those whose rank is equal to $k$ or $k-1$. So when we end up selecting $t$ blocks of circulant matrices for a given $(m, p)$ combination, we are selecting $t$ blocks from the set of rank $k$ matrices and an additional $t$ blocks from the rank $k-1$ matrices. After horizontally joining these matrices we  construct two QT codes with lengths and dimensions $[tm, k]$ and $[tm, k-1]$ and check their minimum distances against the BKLC for those parameters.

In the next table, we give a few data points about the rank distribution for $N = 924$ over $GF(5)$. We give the values of $m \cdot p$ on the left and then the following columns are the rank and the number of matrices of that rank.

\begin{table}
\caption{the table below just shows a small part of all ranks distribution}\label{my-label}
\resizebox{0.85\textwidth}{!}{\begin{minipage}{\textwidth}
\begin{tabular}{ccccc}
\hline\noalign{\smallskip}
$N = m \cdot p$ &Rank1: No.&Rank2: No.&Rank3: No.&Rank4: No.\\
\noalign{\smallskip}\hline\noalign{\smallskip}
$924 = 4 \cdot 231$ & 3 : 49 & 2 : 140 & 1 : 21 & n/a \\
$924 = 6 \cdot 154$ & 4 : 126 & 3 : 14 & 2 : 14 & n/a \\
$924 = 21 \cdot 44$ & 7 : 32 & 6 : 12 & n/a & n/a \\
$924 = 22 \cdot 42$ & 11 : 28 & 10 : 14 & n/a & n/a \\
$924 = 28 \cdot 33$ & 14 : 20 & 13 :9 & 13 : 4 & n/a \\
$924 = 12 \cdot 77$ & 5 : 50 & 4 : 18 & 3 : 6 & 2 : 3 \\
$924 = 4 \cdot 231$ & 2 : 151 & 1 : 59 & n/a & n/a \\
$924 = 12 \cdot 77$ & 4 : 70 & 2 : 6 & n/a & n/a \\
$924 = 11 \cdot 84$ & 11 : 56 & 10 : 28 & n/a & n/a \\
$924 = 14 \cdot 66$ & 13 : 48 & 12 : 18 & n/a & n/a \\
\noalign{\smallskip}\hline
\end{tabular}
\end{minipage}}
\end{table}


If we choose the highest rank matrices, there is an interesting property of the $\gcd$ of all possible generator polynomials with $x^m-a$ and the $\gcd$ of all generator polynomials that corresponds to the highest rank circulant matrices with $x^m-a$. Given the context above, we have 
\begin{theorem} Let $B=\{g_1(x),\dots,g_p(x)\}$ be a set of generators of constacyclic codes, and let  $A=\{f_1(x),\dots, f_t(x)\}$ be the subset of $B$ consisting of those polynomials that correspond to circulant matrices of highest rank.  Then we have
$\gcd(A,x^m-a)=\gcd(B,x^m-a)$

\end{theorem}
\begin{proof}

Let $B = \{g_{1}(x), \cdots g_{p}(x)\}$ be the set of $p$ defining polynomials as in the algorithm. Let $g(x) = \gcd \{B, x^m-a\}$ and let $k = m - \deg(g(x))$. Thus for every $ g_{i}(x) \in B$, we may rewrite $g_{i}(x) = g(x) \cdot u_{i}(x)$ for some $u_{i}(x) \in \Fq[x]$. Let $A = \{f_{1}(x), \dots, f_{t}(x)\}$ be the subset of $B$ such that the constacyclic code generated by each $f_{i}(x)$ has dimension $k$. Thus for every $f_{i}(x) \in A$, we may rewrite $f_{i}(x) = g(x) \cdot v_{i}(x)$ for some $v_{i}(x) \in \Fq[x]$ such that $\gcd(v_{i}(x), x^m-a) = 1$, for if this were not true then the constacyclic code generated by $f_{i}(x)$ would not have dimension $k$. Thus $\gcd(A, x^m-a) = g(x) = \gcd(B, x^m-a)$

\end{proof}

Thus, the codes obtained by our algorithm are of the form $(g_1(x),\dots,g_t(x))$ where $g_1(x),\dots,g_t(x)$ correspond to the highest rank matrices. We can find a lower bound on the minimum distance, and an equivalence between the ASR algorithm and this algorithm.
\begin{theorem} \cite{qtmain}
Let $\gcd(g_1(x),\dots g_t(x),x^m-a)$=D(x) and $f_i(x)=\frac{g_i(x)}{D(x)}$ for $i=1,2, \dots t$.
Then the 1-generator QT code $C$ generated by $(g_1(x),\dots ,g_t(x))$ is of length $m\cdot t$, dimension $m-\deg(D(x))$, and $d(C)\geq t\cdot d$, where $d$ is the minimum distance of the constacyclic code generated by $D(x)$.
\end{theorem} 
In the ASR algorithm, we start with a generator  $g(x)$, which corresponds to the $D(x)$ above. Then by finding $q_i(x)$ that is co-prime with $h(x)=\frac{x^m-a}{g(x)}$, which corresponds to the $f_i(x)$ above, we form  QT codes with generators of the form $(g(x)q_1(x),g(x)q_2(x),\dots,g(x)q_t(x))$. Hence, a code constructed by the top down method of the form $(D(x)f_1(x),D(x)f_2(x),\dots, D(x)f_t(x))$ is essentially the same as a code constructed by the ASR algorithm.


So for each  possible value $t$ for the number of blocks, we want to construct  QT codes of length $t\cdot m$  and dimensions $k$ or $k-1$. These $t$ blocks come from the $p$ circulant matrices we have already constructed, and we  only join circulant matrices of the same rank. Clearly the total number of ways we can do this is $\binom{p}{t}$, which in general is very large so we have imposed a limit of $20,000$. The next step is to randomly select $t$ circulant matrices of rank $k$, and $t$ of rank $k-1$, horizontally join them, and construct two codes with parameters $[t\cdot m, k]_{q}$ and $[t\cdot m, k-1]_{q}$. Finally we compute the  minimum distance of each code and  compare it against the BKLC in \cite{database}.


Take our record breaking code $[84,19,41]_{5}$ as an example. We chose $N=840$, $q=5$ and $a=1$. Firstly, we need to find out the list of possible generator polynomials $g(x)$, which are divisors of $x^{840}-1$. Using an original generator polynomial of degree $765$ (so the dimension of the original constacyclic code is $75$), we found all possible $m$'s and $p$'s. The values of $m = 21$ and $p = 40$ yielded a record breaker with $t = 4$ blocks. In this case, of the 40 blocks in total,  we chose $4$ blocks from the $35$ blocks with the highest rank, which is $19$. Each of the remaining $5$ blocks has rank  $18$. The  dimension is determined  by computing $k = m - \deg(\gcd(x^m-a, g_{1}(x), g_{2}(x), \dots, g_{p}(x))$). After we  split up the original long generator polynomial using our selected values of $m$ and $p$, we construct a circulant matrix for each of the $p = 40$ polynomials. Once we have computed $k$, we go through all of these circulant matrices and only keep those with rank $k$ or $k-1$. After selecting our $t = 4$ full rank matrices, we horizontal join them, and get a final matrix of  $84=t\cdot m=4\cdot 21$ columns with rank $19$. For a particular choice of these 4 blocks, we obtained a new linear code whose minimum distance is 41, improving the minimum distance of the previously BKLC of this length and dimension given in \cite{database}. The  defining (generating) polynomials of each block of this code are given in the table below.


\section{New Codes}

We found two types of new codes from an implementation of the algorithm. We have found 5 QT codes that are new among all linear codes according to the database \cite{database}.  They are listed in the table below.

\begin{table}[h]
\resizebox{0.68\textwidth}{!}{\begin{minipage}{\textwidth}
\caption{Record breaking QT codes}

\label{tab-1}
\begin{tabular}{p{.5cm}p{2cm}p{1cm}p{1.5cm}p{1cm}p{5cm}}
\hline\noalign{\smallskip}
&$[n,k,d]_{q}$  & $\alpha$ & $N$ & $m$ & Polynomials  \\
\noalign{\smallskip}\hline\noalign{\smallskip}
    1 & $[85, 16, 45]_{5}$ &1 & 2142 & 17 &$g_{1}$= [1331311000103332] \\
  & & & & & $g_{2}$= [2030141001241411]\\ 
  & & & & & $g_{3}$= [3412143022013031]\\
  & & & & & $g_{4}$= [1123200013130012]\\
  & & & & & $g_{5}$= [4233002104312041]\\
  [.5ex] 
      2 & $[84, 19, 41]_{5}$ &1 & 840 & 21 &$g_{1}$= [44010311002304111222] \\
  & & & & & $g_{2}$= [141032024233443231]\\ 
  & & & & & $g_{3}$= [1310443344442044020]\\
  & & & & &$g_{4}$= [4011113222243123010]\\
  [.5ex] 
      3 & $[84, 13, 48]_{5}$ &1 & 3276 & 14 &$g_{1}$= [34030034422424] \\
  & & & & & $g_{2}$= [4023414111414]\\ 
  & & & & & $g_{3}$= [10342023404034]\\
  & & & & & $g_{4}$= [44300024241344]\\
  & & & & & $g_{5}$= [3221030030213]\\
  & & & & & $g_{6}$= [433111043203]\\
  [.5ex] 
  
      4 & $[65, 12, 39]_{7}$ &1 &35100&13 &$g_{1}$= [2322660251501] \\
  & & &&& $g_{2}$= [4415215004556]\\ 
  &&& & & $g_{3}$= [1551620013551]\\
  & &&& & $g_{4}$= [4030032120616]\\
  & & &&& $g_{5}$= [4626364150311]\\
  [.5ex] 
  5 & $[78, 13, 47]_{7}$ &1 &4680&13&$g_{1}$= [6536106450546] \\
  & & &&& $g_{2}$= [640410515651]\\ 
  &&& & & $g_{3}$= [32251003506]\\
  & &&& & $g_{4}$= [524220205542]\\
  & & &&& $g_{5}$= [520330333466]\\
  &&&&& $g_{6}$= [15211116242]\\
  [.5ex] 
\noalign{\smallskip}\hline
\end{tabular}
\end{minipage}}

\end{table}

The second type of  codes we have found are new in the class of QC and QT codes as presented in the online database \cite{qttables}. Often times, these codes have the parameters of the BKLCs given in \cite{database} and at the same time have more simple construction than the BKLCs in the database. For example, consider the code with parameters $[66,11,40]_7$ in the table below (number 4). It has the same parameters as the BKLC given in \cite{database} but the construction of the code in the database is indirect and complicated, involving many steps and manipulations of other codes. The code that we have is QC, so it has a more useful and desirable construction. We have found a large number of such codes. The table below contains 40 of them. In this table,
$\star$ delineates a new QC code for the online database of QT codes \cite{qttables},
$\star\star$ delineates a new QT code for the online database of QC codes \cite{qttables}, and
$\circ$ delineates a code with better construction than what is given in \cite{database}

\begin{table}

\resizebox{0.65\textwidth}{!}{\begin{minipage}{\textwidth}
\begin{tabular}{cccccc}
\hline\noalign{\smallskip}
&$[n,k,d]_{q}$  & $\alpha$ & $N$ & $m$ & Polynomials  \\
\noalign{\smallskip}\hline\noalign{\smallskip}
    1 & $[75,16,41]_{7}$ * &1&29700&25 &$g_{1}$= [263433044421333266145152] \\
  & & &&& $g_{2}$= [315432322306666040522014]\\ 
  & & &&& $g_{3}$= [423226456521644016532614]\\
  [.5ex] 
      2 & $[70,10,45]_{7}$ * &1&1530&10 &$g_{1}$= [6303410364] \\
  & & &&& $g_{2}$= [2050420624]\\ 
  & & &&& $g_{3}$= [2210321115]\\
  & & &&& $g_{4}$= [45041324]\\
  & & &&& $g_{5}$= [435123455]\\
  & & &&& $g_{6}$= [231612611]\\
  & & &&& $g_{7}$= [651163446]\\
  [.5ex] 
      3 & $[50,12,28]_{7}$ * &1&1650&25 &$g_{1}$= [540452551353554052546141] \\
  & & &&& $g_{2}$= [641633425035353011331645]\\ 
  [.5ex] 
      4 & $[66, 11, 40]_{7}$ *$^\circ$ &1&1650&11 &$g_{1}$= [2322660251501] \\
  & & &&& $g_{2}$= [4415215004556]\\ 
  & & &&& $g_{3}$= [1551620013551]\\
  & & &&& $g_{4}$= [4030032120616]\\
  & & &&& $g_{5}$= [4626364150311]\\
  [.5ex] 
      5 & $[77,11,49]_{7}$ * &1&1650&11 &$g_{1}$= [564114442] \\
  & & &&& $g_{2}$= [2461136364]\\ 
  & & &&& $g_{3}$= [1516121646]\\
  & & &&& $g_{4}$= [6616204266]\\
  & & &&& $g_{5}$= [1645256106]\\
  &&&&& $g_{6}$= [261414153]\\
    &&&&& $g_{7}$= [5230445332]\\
  [.5ex] 
        6 & $[84, 11, 54]_{7}$ *$^\circ$ &1&2340&14 &$g_{1}$= [444322262213] \\
  & & &&& $g_{2}$= [554352030341]\\ 
  & & &&& $g_{3}$= [10166310604]\\
  & & &&& $g_{4}$= [15346150532]\\
  & & &&& $g_{5}$= [61316462321]\\
    &&&&& $g_{6}$= [32314565553]\\
    &&&&& $g_{7}$= [15302241541]\\
  [.5ex] 
      7 & $[63,21,27]_{7}$ *&1 &3276&21&$g_{1}$= [124023403431134343212] \\
  & & &&& $g_{2}$= [213132232321244113304]\\ 
  & & &&& $g_{3}$= [43103444244422211142]\\
  [.5ex] 
        8 & $[75,15,42]_{7}$ *&1&4680&15 &$g_{1}$= [2322660251501] \\
  & & &&& $g_{2}$= [4415215004556]\\ 
  & & &&& $g_{3}$= [1551620013551]\\
  & & &&& $g_{4}$= [4030032120616]\\
  & & &&& $g_{5}$= [4626364150311]\\
  [.5ex] 
      9 & $[76,10,50]_{7}$ *$^\circ$&1&5586&19 &$g_{1}$= [5031221102660240321] \\
  & & &&& $g_{2}$= [2363456125525046026]\\ 
  & & &&& $g_{3}$= [261352230616301205]\\
  & & &&& $g_{4}$= [645316064010314463]\\
  [.5ex] 
        10 & $[60,19,27]_{7}$ *$^\circ$&1&7020&20 &$g_{1}$= [15443624352620433401] \\
  & & &&& $g_{2}$= [50055310242611250602]\\ 
  & & &&& $g_{3}$= [2424542451060242561]\\
  [.5ex] 
      11 & $[88,10,59]_{7}$ * &1&9900&10 &$g_{1}$= [5203436246] \\
  & & &&& $g_{2}$= [51043451221]\\ 
  & & &&& $g_{3}$= [12401350525]\\
  & & &&& $g_{4}$= [41243403]\\
  & &&& & $g_{5}$= [53444636644]\\
  &&&&& $g_{6}$= [1031143125]\\
    &&&&& $g_{7}$= [33452534051]\\
  &&&&& $g_{8}$= [22465515122]\\
  [.5ex] 
  12 & $[60,22,24]_{7}$ *$^\circ$ &1 &23400&30 &$g_{1}$= [36635456033343033602466243345] \\
  & & &&& $g_{2}$= [656321204662143163155634466551]\\ 
  [.5ex] 
  13 & $[66,22,28]_{7}$ *$^\circ$&1&23760&22 &$g_{1}$= [6563251503162345354255] \\
  &&& & & $g_{2}$= [533056566022662664544]\\ 
  & &&& & $g_{3}$= [54064355203553654613]\\
  [.5ex] 
  14 & $[50,20,20]_{7}$ * &1&23400&25 &$g_{1}$= [4443213452310165035153651] \\
  & & &&& $g_{2}$= [1630165556635256134420642]\\ 
  [.5ex] 
      15 & $[78, 23, 33]_{5}$ **  &2 & 1638 & 39 &$g_{1}$= [13032111444411404203244031330204130] \\
  & & & & & $g_{2}$= [424210422221110200143114323340030334]\\ 
  [.5ex] 
  
16 & $[84, 13, 47]_{5}$ *$^\circ$  &1 & 2394 & 42 &$g_{1}$= [130321114444114042032440313302041302] \\
  & & & & & $g_{2}$= [424210422221110200143114323340030334]\\ 
  [.5ex] \\
  17  & $[63, 14, 32]_{5}$ *$^\circ$  &1 & 2394 & 21 &$g_{1}$= [321343402000310324] \\
  & & & & & $g_{2}$= [34343420432104113431]\\ 
  & & & & & $g_{3}$= [11021044124132020034]\\
  [.5ex] 
  
   18 & $[70, 13, 38]_{5}$ *$^\circ$  &1  & 2520 & 14 &$g_{1}$= [14442024404443] \\
  & & & & & $g_{2}$= [4003124334114]\\ 
  & & & & & $g_{3}$= [31142014342]\\
  & & & & & $g_{4}$= [4233234211104]\\
  & & & & & $g_{5}$= [4133124411213]\\
  [.5ex] 
\noalign{\smallskip}\hline
\end{tabular}
\end{minipage}}
\end{table}

\begin{table}
\vspace{-.3cm}

\resizebox{0.67\textwidth}{!}{\begin{minipage}{\textwidth}
\begin{tabular}{cccccc}
\hline\noalign{\smallskip}
&$[n,k,d]_{q}$  & $\alpha$ & $N$ & $m$ & Polynomials  \\
\noalign{\smallskip}\hline\noalign{\smallskip}

  19 & $[57, 10, 33]_{5}$ *  &1 & 2394 & 19  &$g_{1}$= [4001122202200132343] \\
  & & & & & $g_{2}$= [2001010134034202224]\\ 
  & & & & & $g_{3}$= [423311200332331241]\\
  [.5ex] 
  20 & $[54, 16, 24]_{5}$ *  &1 & 2520 & 18 & $g_{1}$= [111112401404020422] \\
  & & & & & $g_{2}$= [210103420211004144]\\ 
  & & & & & $g_{3}$= [430443201130124111]\\
    21 & $[57, 18, 24]_{5}$ *  &1 & 2394 & 19 &$g_{1}$= [12023114004324224] \\
  & & & & & $g_{2}$= [31032100304014224]\\ 
  & & & & & $g_{3}$= [140212320140240333]\\
  [.5ex] 
  22 & $[70, 12, 39]_{5}$ *  &1 & 2394 & 14 &$g_{1}$= [22011333323034] \\
  & & & & & $g_{2}$= [2122232302344]\\ 
  & & & & & $g_{3}$= [23043423422123]\\
    & & & & & $g_{4}$= [14114212211]\\
  & & & & & $g_{5}$= [3130431234204]\\
  [.5ex] 
    23 & $[57, 19, 23]_{5}$ *  &1 & 2394 & 19 &$g_{1}$= [21201010102211234] \\
  & & & & & $g_{2}$= [24014322314212313]\\ 
  & & & & & $g_{3}$= [120021024011344434]\\
  [.5ex] 
  24 & $[54, 17, 23]_{5}$ *  &1 & 2394 & 18 &$g_{1}$= [111142143224042044] \\
  & & & & & $g_{2}$= [13203240213132332]\\ 
  & & & & & $g_{3}$= [43013432321002133]\\
  [.5ex] 
  25 & $[54, 15, 25]_{5}$ *  &1 & 2394 & 18 &$g_{1}$= [21201010102211234] \\
  & & & & & $g_{2}$= [24014322314212313]\\ 
  & & & & & $g_{3}$= [120021024011344434]\\
  [.5ex] 
  26 & $[48, 21, 16]_{5}$ * $^{\circ}$  &1 & 840 & 24 &$g_{1}$= [43220423003100340303233] \\
  & & & & & $g_{2}$= [1420414002032004203044]\\ 
  [.5ex] 
  27 & $[63, 21, 25]_{5}$ *$^\circ$  &1 & 1638 & 21 &$g_{1}$= [301304303123033402013] \\
  & & & & & $g_{2}$= [143331010304032042133]\\ 
  & & & & & $g_{3}$= [22102000202342013201]\\
  [.5ex] 
  28 & $[48, 20, 17]_{5}$ *$^\circ$  &1 & 840 & 24 &$g_{1}$= [4112013131321102321333] \\
  & & & & & $g_{2}$= [20221032200410113034424]\\ 
  [.5ex] 
  29 & $[48, 19, 18]_{5}$ *$^{\circ}$  &1 & 840 & 24 &$g_{1}$= [4132313101410322242343] \\
  & & & & & $g_{2}$= [2234233102121102440302]\\ 
  [.5ex] 
  30 & $[98, 14, 56]_{5}$ *  &1 & 840 & 14 &$g_{1}$= [30434340014012] \\
  & & & & & $g_{2}$= [404230140102]\\ 
  & & & & & $g_{3}$= [12113334213201]\\
  & & & & & $g_{4}$= [3042242213423]\\
  & & & & & $g_{5}$= [4144330200204]\\
  & & & & & $g_{6}$= [1400023210111]\\
  & & & & & $g_{7}$= [1140211211241]\\
  [.5ex] \\
  31 & $[54, 18, 22]_{5}$ *  &1 & 2520 & 18 &$g_{1}$= [104230420440014] \\
  & & & & & $g_{2}$= [1002400240424024]\\ 
  & & & & & $g_{3}$= [440320311201032]\\
  [.5ex] 
  32 & $[54, 16, 24]_{5}$ *  &1 & 2394 & 18 &$g_{1}$= [213013442013241243] \\
  & & & & & $g_{2}$= [343444014210104]\\ 
  & & & & & $g_{3}$= [1141401024010033]\\
  [.5ex] 
  33 & $[63, 20, 26]_{5}$ *$^\circ$  &1 & 840 & 21 &$g_{1}$= [433230444340230011144] \\
  & & & & & $g_{2}$= [330024240231414133442]\\ 
  & & & & & $g_{3}$= [30222313400013211232]\\
  [.5ex] 
  34 & $[63, 19, 27]_{5}$ *$^\circ$  &1  & 1638 & 21 &$g_{1}$= [234410410332004300314] \\
  & & & & & $g_{2}$= [14213013022200423341]\\ 
  & & &  & & $g_{3}$= [20003430124131132214]\\
  [.5ex] 
  35 & $[63, 18, 28]_{5}$ *$^\circ$  &1 & 1638 & 21 &$g_{1}$= [212000340143014102304] \\
  & & & & & $g_{2}$= [10040010022420444214]\\ 
  & & & & & $g_{3}$= [11432313222320030332]\\
  [.5ex] 
  36 & $[63, 15, 31]_{5}$ *  &1 & 2310 & 21 &$g_{1}$= [333410132012211231111] \\
  & & & & & $g_{2}$= [1341022300043442024]\\ 
  & & & & & $g_{3}$= [300211014411004001]\\
  [.5ex] 
  37  & $[84, 15, 45] _{5}$ *  &1 & 2310 & 21 &$g_{1}$= [342412033421313430222] \\
  & & & & & $g_{2}$= [100211424324222410432]\\ 
  & & & & &$g_{3}$= [223300414213401242213]\\
    & & & & & $g_{4}$= [11214331334014103314]\\
  [.5ex] 
  38 & $[72, 22, 30]_{5}$ *  &1 & 2520 & 24 &$g_{1}$= [30204234204423333414224] \\
  & & & & & $g_{2}$= [20030300341134112302]\\ 
  & & & & & $g_{3}$= [1203004014301110344232]\\
  [.5ex] 
  39 & $[48, 16, 21]_{5}$ *  &1 & 6552 & 24 &$g_{1}$= [30131044301202334202204] \\
  & & & & & $g_{2}$= [34041432001111423440011]\\ 
  [.5ex] 
  40 & $[52, 14, 25]_{5}$ **  &2 & 9828 & 26 &$g_{1}$= [422412211222002300041024] \\
  & & & & & $g_{2}$= [4112244244422042111422143]\\ 
  [.5ex] \\
  
\noalign{\smallskip}\hline
\end{tabular}
\end{minipage}}
\end{table}

   
\clearpage



%

\end{document}